\documentclass[11pt]{article}
\usepackage{amsmath}
\usepackage{lscape}
\usepackage{amssymb}
\usepackage{color}
\usepackage[nice]{nicefrac}
\usepackage{graphicx}
\newtheorem{theorem}{Theorem}

\newtheorem{corollary}[theorem]{Corollary}

\newtheorem{definition}[theorem]{Definition}

\newtheorem{proposition}[theorem]{Proposition}
\newtheorem{remark}[theorem]{Remark}

\newenvironment{proof}[1][Proof]{\noindent\textbf{#1.} }{\ \rule{0.5em}{0.5em}}

\usepackage[top=1.2in, bottom=1.0in, left=1in, right=1in]{geometry}

\begin{document}

\title{Statistical Arbitrage in the Black-Scholes Framework}
\author{Ahmet G\"{o}nc\"{u}\thanks{Corresponding author. Email: Ahmet.Goncu@xjtlu.edu.cn}\\
Department of Mathematical Sciences\\
Xian Jiaotong Liverpool University, Suzhou, 215123, China.\\
Affiliation: Bogazici University, Center for Economics and Econometrics, Istanbul, Turkey.\\
Ahmet.Goncu@xjtlu.edu.cn}

\maketitle
\abstract
In this study we prove the existence of statistical arbitrage opportunities in the Black-Scholes framework by considering trading strategies that consists of borrowing from the risk free rate and taking a long position in the stock until it hits a deterministic barrier level. We derive analytical formulas for the expected value, variance, and probability of loss for the discounted cumulative trading profits. Statistical arbitrage condition is derived in the Black-Scholes framework, which imposes a constraint on the Sharpe ratio of the stock. Furthermore, we verify our theoretical results via extensive Monte Carlo simulations.\\
\\
\textbf{Statistical arbitrage, Black-Scholes model, Geometric Brownian motion, Monte Carlo simulation. }

\section{Introduction}\label{Section:Intro}
Investment communities consider statistical arbitrage as the mispricing of any security according to their expected future trading value in relationship with their spot prices. Statistical arbitrage strategies originally evolved from the so called ``pairs trading'', which exploits the mean reversion in the performance of a pair of stocks identified based on various criteria. Amongst others \cite{Do:2009}, \cite{Gatev:2006}, \cite{Cummins:2012}, \cite{Elliot:2005}, and \cite{Avellaneda:2010} investigate the performance of pairs trading and statistical arbitrage strategies. Optimal statistical arbitrage trading for an Ornstein-Uhlenbeck process is given in \cite{Bertram:2010}. Mathematical definitions for statistical arbitrage strategies are given in the studies by \cite{HoganEtAl:2004}, \cite{Jarrow:2005}, \cite{Jarrow:2012}, and \cite{Bondarenko:2003}. \cite{Bondarenko:2003} assumes the existence of a derivatives markets, however, in this study we do not have such an assumption. Using the definition of statistical arbitrage and with some additional assumptions on the dynamic behavior of statistical arbitrage profits, hypothesis tests for the existence of statistical arbitrage are derived in \cite{HoganEtAl:2004}, \cite{Jarrow:2005}, and \cite{Jarrow:2012}. These hypothesis tests are used to test the existence of statistical arbitrage and efficiency of the market, which avoids the joint hypothesis problem stated in \cite{Fama:1998}. 

In the study by \cite{HoganEtAl:2004} a mathematical definition for statistical arbitrage is given with various examples. Following the definition of \cite{HoganEtAl:2004} and considering the \cite{BS:1973} model, where stock prices follow geometric Brownian motion process, we present examples of statistical arbitrage strategies and prove the existence of statistical arbitrage opportunities. 

If an investor has better information (compared to the market) to identify the stocks with high (or low) expected growth rates, then there exists statistical arbitrage opportunities in the Black-Scholes framework of stock price dynamics. In this paper we present trading strategies that yield statistical arbitrage in the Black-Scholes model and then derive a no-statistical arbitrage condition. The derived no-statistical arbitrage condition imposes a constraint on the Sharpe ratio of stocks. If an investor knows or believes that he knows the stocks that satisfy the statistical arbitrage condition, then this is sufficient to design statistical arbitrage trading strategy.

This article is organised as follows. In the next section we present the definition of statistical arbitrage and provide examples of statistical arbitrage strategies. In Section \ref{Section:StatArbBSModel}, we prove the existence of statistical arbitrage in the Black-Scholes framework and derive the no-statistical arbitrage condition. In Section \ref{Section:OtherProperties}, we present some other properties of statistical arbitrage strategies. We conclude in Section \ref{Section:Conclusion}.

\section{Statistical Arbitrage}\label{Section:StatisticalArbitrage}

Given the stochastic process for the discounted cumulative trading profits, denoted as $\{v(t):t\geq 0\}$ that is defined on a probability space $(\Omega, \mathcal{F}, {P})$ the statistical arbitrage is defined as follows (see \cite{HoganEtAl:2004}).
\begin{definition}\label{Def:StatisticalArbitrage}
A statistical arbitrage is a zero initial cost, self-financing trading strategy $\{v(t):t\geq 0\}$ with cumulative discounted value $v(t)$ such that
\begin{enumerate}
\item $v(0) = 0$
\item $\lim_{t\to \infty} E[v(t)] > 0$,
\item $\lim_{t\to \infty}{{P}(v(t)<0)}=0$, and
\item $\lim_{t\to \infty}{\frac{var(v(t))}{t}}=0$ if ${P}(v(t)<0)>0,\quad$ $\forall t<\infty$. 
\end{enumerate}
\end{definition}

\cite{HoganEtAl:2004} states that the fourth condition applies only when there always exists a positive probability of losing money. Otherwise, if the probability of loss becomes zero in finite time $T$, i.e. $P(v(t)<0)=0$ for all $t\geq T$, this implies the existence of a standard arbitrage opportunity. 

A standard arbitrage opportunity is a special case of statistical arbitrage. Indeed, for a standard arbitrage strategy $V$ (self-financing) there exists a finite time $T$ such that $P(V(t)>0)>0$ and $P(V(t)\geq 0)=1$ for all $t\geq T$ and the proceeds of this profit can be deposited into money market account for the rest of the infinite time horizon. This gives $V(s)=V(T)B_s/B_T$ for $s\geq t$. The discounted value of this strategy is given by $v(s) = V(T)(B_s/B_T)(1/B_s)=v(T)$, which satisfies Definition \ref{Def:StatisticalArbitrage}. \\
\\
\textbf{Example 1: Black-Scholes example of \cite{HoganEtAl:2004}}\label{Ex:1} ($\alpha>r_f$): 

Consider the standard Black-Scholes dynamics for stock price (non-dividend paying) $S_t$ that evolves according to 
\begin{equation}
S_t = S_0 \exp((\alpha-\sigma^2/2)t + \sigma dW_t), 
\end{equation}
where $W_t$ is standard Brownian motion process, $\alpha$ is the growth rate of the stock price, $r_f$ is the risk-free rate, and $\sigma$ is the volatility, which are assumed to be constant. Following \cite{HoganEtAl:2004} we consider $\alpha>r_f$. Money market account follows $B_t = \exp(r_f t)$. \cite{HoganEtAl:2004} considers a self-financing trading strategy that consists of buying and holding one unit of stock financed by the money market account with constant risk-free rate $r_f$. 

The value of the cumulative profits at time $t$ is 
\begin{equation}
V(t) = S_t - S_0 e^{r_f t}=S_0( e^{(\alpha-\sigma^2/2)t +\sigma W_t } - e^{r_f t}),
\end{equation}
whereas the discounted cumulative value of the trading profits is given as 
\begin{equation}
v(t) = S_0 \left( \exp( (\alpha-\sigma^2/2 - r_f)t + \sigma W_t) -1 \right).
\end{equation}
Since $W_t \sim N(0,t)$ for each $t$ we conclude that

\begin{equation}
E[v(t)] = S_0 \left( e^{(\alpha-r_f)t } -1 \right)
\end{equation}
and $\lim_{t\to \infty} E[v(t)] = \infty$.
We obtain the variance as 

\begin{equation}\label{Eq:Variance1}
var(v(t)) = \left( e^{\sigma^2 t} -1 \right) S_0^2 \left( e^{2(\alpha - r_f)t}\right) \to \infty \hbox{ as } t\to \infty.
\end{equation}
Therefore, Condition 4 in Definition \ref{Def:StatisticalArbitrage} is not satisfied. \\
\\
Using the above example, Hogan et al. (2004) concludes that the Black-Scholes model excludes statistical arbitrage opportunities according to Definition \ref{Def:StatisticalArbitrage}. This example is also used to justify the existence of a fourth condition in the Definition \ref{Def:StatisticalArbitrage}. It is mentioned that without the fourth condition buy and hold strategies yield statistical arbitrage opportunities in the Black-Scholes model for $\alpha-r_f>\sigma^2/2$. In the next example we show that the fourth condition in Definition \ref{Def:StatisticalArbitrage} can be satisfied with a different trading strategy. But first, as a natural result of Example 1, we state the following proposition. 

\begin{proposition}\label{Proposition1}
For all the buy and hold trading strategies consisting of a single stock in the Black-Scholes model, the time averaged variance of the discounted cumulative trading profits goes to infinity, i.e. $\lim_{t\to\infty}var(v(t))/t =\infty$, for $\alpha-r_f>-\sigma^2/2$, and it decays to zero for $\alpha-r_f\leq -\sigma^2/2$.
\end{proposition}

\begin{proof}
The long position in the risky asset has present value $e^{-r_f t}S_t$. Since the money market account is deterministic, the variance of our trading profits depends only on the investment in the risky asset $S_t$. Hence, the variance of the random component $S_t$, which is given in Equation \ref{Eq:Variance1} always has the exponential term in the order of $e^{2(\frac{\sigma^2}{2} + \alpha - r_f)t}$. Then clearly $\lim_{t\to\infty}{var(v(t))/t}=\infty$ for $\alpha-r_f > -\sigma^2/2$ and decays to zero for $\alpha-r_f\leq -\sigma^2/2$.
\end{proof}

For $\alpha-r_f\leq-\sigma^2/2$, we can create statistical arbitrage by short selling the stock and investing in the money market account at time 0 and renewing the short position and re-investing profits (or re-financing losses) in the money market account. Since the expected cumulative discounted trading profits converge to $S_0$ and the time averaged variance decays to zero, this yields statistical arbitrage.
 
It is important to note that to create statistical arbitrage opportunities for stocks with sufficiently large positive expected growth rates, we need to impose a stopping or selling condition on the trading strategy. Without a stopping boundary, we keep holding the stock and by Proposition \ref{Proposition1} we fail to satisfy the condition that the time average of the variance must decay to zero. 

If we successfully introduce this selling or stopping condition in our trading strategies, we profit from the positive expected growth in the stock price, but at the same time we can reduce the holding of the risky asset (reduce the holding of the risky asset to zero in time) and control the time averaged variance. Next, we present this idea with an example.\\
\\
\textbf{Example 2:}\label{Ex:2} 
We introduce a termination condition for the buy and hold strategy as follows: whenever the stock price process hits to a constant barrier level sell the stock and invest in the money market account. In this way, we utilize the finiteness of the first passage time of the Brownian motion process. 

The discounted cumulative trading profits in our ``buy and hold until barrier'' strategy is given by 
\begin{equation}
v(t) = \begin{cases}
Be^{-r_f t^*} - S_0, & \hbox{ if } t^*\leq t \\
S_t e^{-r_f t} - S_0, & \hbox{ else},
\end{cases}
\end{equation}
where $t^* = \min\{t\geq 0: S_t = B\}$ and $B>S_0$ is the constant barrier level. 

If the stock price hits the barrier level at infinity then the trading loss is $S_0$, whereas if it hits in finite time, then we have $Be^{-r_f t^*}-S_0$ as the discounted value of our trading strategy. \\
\\
\\
Consider the Brownian motion process with drift given by 
\begin{equation}\label{Eq:BMwithDrift}
X_t = \mu t + W_t, 
\end{equation}
where $W_t$ is the standard Brownian motion, $\mu\in\Re$, and denote the first passage time of this process as
\begin{equation}
\tau_m = \min\{t\geq 0: X(t) = m\}
\end{equation}
for fixed $m$. 

The Brownian motion with drift, $X_t$, hits the level $m$ in finite time almost surely for $\mu>0$, i.e. $P(\tau_m<\infty)=1$. The Laplace transform of the first passage time for $X_t$ is equal to (see \cite{Shreve:2004} page 120)
\begin{equation}\label{Eq:ShreveFormula}
E[e^{-r \tau_m}] = e^{m\mu - m\sqrt{2r_f + \mu^2}}, \quad \mbox{ for all } r>0.
\end{equation}
In Equation \ref{Eq:ShreveFormula}, put $X_t = \ln(S_t/S_0)/\sigma$, $\mu=(\alpha-\sigma^2/2)/\sigma$ and $m=\ln(B/S_0)/\sigma$. Then the process for $X_t=\mu t + W_t$ is the same as 
\begin{equation}
S_t = S_0 \exp((\alpha-\sigma^2/2)t + \sigma W_t), 
\end{equation}
as it is under the Black-Scholes model. \\
\\
\textbf{Expected value of the trading profits is positive (if $\alpha>r_f$):}\\ 
Note that the right hand side of Equation \ref{Eq:ShreveFormula} can be equivalently written as
$\left(\frac{B}{S_0}\right)^{(\mu - \sqrt{2r_f + \mu^2})/\sigma}$. The result in Equation \ref{Eq:ShreveFormula} yields the following formula for the expected trading profits for sufficiently large $t$:
\begin{equation}\label{Eq:ExpectedValue}
E[v(t)] = B \left(\frac{B}{S_0}\right)^{(\mu - \sqrt{2r_f + \mu^2})/\sigma}  - S_0
\end{equation}
with positive expected trading profits
\begin{equation}\label{Eq:Profit}
\lim_{t\to\infty}{E[v(t)]} >0,
\end{equation}
for $\mu>0$ and $B>S_0$. 

\begin{proposition}
If $\mu>0$ and $m>0$, then the limit of the expected discounted profits in the trading strategy given in Example \ref{Ex:2} is always positive:
\begin{equation}\label{Eq:Profit2}
\lim_{t\to\infty}{E[v(t)]} =  B e^{m\mu - m\sqrt{2r_f + \mu^2}} - S_0>0.
\end{equation} 
\end{proposition}
\begin{proof}
Consider the stochastic process given by $$X(t) = \mu t + W(t),$$ where $W(t)$ is standard Brownian motion process and we have $X(t)$, $\tau_m$, and $m$ as defined before. Following the similar steps as given in \cite{Shreve:2004} page 111, we start by writing the martingale process $Z(t)$ as 
\begin{equation}
Z(t) = \exp(\sigma X(t) - (\sigma \mu+\sigma^2/2)t ), 
\end{equation}
which is clearly an exponential martingale with $Z(0)=1$ and utilizing this fact we obtain 
\begin{equation}
E\left[\exp(\sigma X(t\wedge \tau_m) - (\sigma \mu+\sigma^2/2)(t\wedge \tau_m))\right]=1,\quad t\geq 0.
\end{equation}
For $\sigma>0$ and $m>0$ we know that $0\leq \exp(\sigma X(t\wedge \tau_m))\leq e^{\sigma m}$. If $\tau_m<\infty$, we have $\exp(-(\sigma\mu+\sigma^2/2)(t\wedge\tau_m)) =\exp(-(\sigma\mu+\sigma^2/2)\tau_m)$ for large enough $t$, whereas if $\tau_m=\infty$, we have $\exp(-(\sigma \mu+\sigma^2/2)(t\wedge \tau_m)) = \exp(-(\sigma \mu+\sigma^2/2) t )$ and the exponential term converges to zero. We can write these two cases together 
\begin{equation}
\lim_{t\to\infty}{\exp(-(\sigma \mu+\sigma^2/2)(t\wedge \tau_m))} = \exp(-(\sigma \mu+\sigma^2/2)\tau_m),
\end{equation}
%\mathbf{1}_{\{\tau_m<\infty\}}\exp(-(\sigma \mu+\sigma^2/2)\tau_m),
% where $\mathbf{1}_{\{.\}}$ is the indicator function. 

We recall that due to Equation \ref{Eq:ShreveFormula} the first passage time is almost surely finite, i.e. $P(\tau_m<\infty)=1$, and for sufficiently large $t$ we have $\exp(\sigma X(t\wedge \tau_m)=\exp(\sigma X(\tau_m))=\exp(\sigma m)=B/S_0$.
%, whereas for $\tau_m=\infty$ we have $\exp(\sigma X(t\wedge\tau_m))=\exp(\sigma X(t))\leq \exp(\sigma m)$.

Writing the product of two exponential terms as 
\begin{equation}
\lim_{t\to\infty}{\exp(\sigma X(t\wedge\tau_m) - (\sigma \mu+\sigma^2/2)(t\wedge \tau_m)} = \exp(\sigma m-(\mu+\sigma^2/2)\tau_m)
\end{equation}
and interchanging the limit and expectation by the dominated convergence theorem, we obtain
\begin{equation}
1 = E[\exp(\sigma m-(\sigma \mu+\sigma^2/2)\tau_m)],
\end{equation}
where $\mu = (\alpha - \sigma^2/2)/\sigma$. This implies
\begin{equation}\label{Eq:1}
E[\exp(-\alpha\tau_m)] = e^{-m\sigma}=\frac{S_0}{B}.
\end{equation}
Note that $\lim_{t\to\infty}v(t) = Be^{-r_f\tau_m}-S_0$ and $\lim_{t\to\infty}E[v(t)]=E[\lim_{t\to\infty}v(t)]=BE[e^{-r_f\tau_m}]-S_0,$ where $B E[e^{-r_f\tau_m}]-S_0 > BE[e^{-\alpha \tau_m}] - S_0 = 0$ by Equation \ref{Eq:1} for $\alpha>r_f>0$ and thus proves the positivity of the discounted cumulative profits. 
\end{proof}\\
\\
\textbf{Variance of the trading profits:}\\
Next we derive the analytical formula for the variance of our trading profits. 
For sufficiently large $t$ we write 
\begin{equation}
var(v(t)) = E[v^2(t)] - E[v(t)]^2, 
\end{equation}
where 
\begin{eqnarray}
\lim_{t\to\infty}E[v^2(t)] & =&  E[\lim_{t\to\infty} v^2(t)]\\
& = & B^2 E[e^{-2 r_f \tau_m}] -2B S_0 E[e^{-r_f \tau_m}]+S_0^2.
\end{eqnarray}
Therefore, the limit of the variance of cumulative discounted trading profits is given as 
\begin{equation}\label{Eq:Variance}
\lim_{t\to\infty}var(v(t)) = B^2\left[\left(\frac{B}{S_0}\right)^{ (\mu-\sqrt{4r_f + \mu^2})/\sigma }-\left(\frac{B}{S_0}\right)^{( 2\mu-2\sqrt{2r_f + \mu^2})/\sigma} \right].
\end{equation}
\\
\textbf{First passage time:}\\
An investor implementing the trading strategy would also be interested in knowing the distribution of the first passage time and timing to sell the stock. If the conditions $\alpha>\sigma^2/2$ and $B>S_0$ are satisfied then the first passage time of the Brownian motion with drift in Equation \ref{Eq:BMwithDrift} to level $m=\log(B/S_0)/\sigma$ is given by
\begin{equation}\label{Eq:InverseGaussian}
\tau_{m}\sim \hbox{IG}\left( \frac{m}{\mu},m^2\right), 
\end{equation}
where IG represents inverse Gaussian distribution\\
\\
\textbf{Change in expected profits with respect to the barrier level:}\\
To see the effect of an increase in the barrier level $B$ on the stochastic discount rate and trading profits, we take the partial derivatives of Equations \ref{Eq:ShreveFormula} and \ref{Eq:ExpectedValue}. We obtain
\begin{equation}
\frac{\partial E[e^{-r_f \tau_B}] }{\partial B} = \frac{(\mu-\sqrt{2 r_f + \mu^2})}{\sigma} \frac{E[e^{-r_f \tau_B}]}{B}<0,
\end{equation}
which is negative as expected, since an increase in the barrier level means we hit the barrier at a later time on average and the present value of one dollar obtained at the hitting time decreases. 

However, the partial derivative of the trading profits (for sufficiently large $t$) with respect to the barrier level $B$ is positive. Differentiating Equation \ref{Eq:ExpectedValue} with respect to $B$ we obtain
\begin{equation}
\frac{\partial {E[v(t)]}}{\partial B} = \frac{(\sigma+\mu - \sqrt{2r_f + \mu^2})}{\sigma} \left(\frac{B}{S_0}\right)^{(\mu-\sqrt{r_f + \mu^2})/\sigma } > 0, \quad \hbox{ for } \alpha>r_f.
\end{equation}

The above results show that the expected profits increase as the barrier level increases when $t$ is sufficiently large, whereas a higher barrier level means that on average we need to wait longer time to observe the variance to be bounded. In other words, while our expected profit increases, we suffer from a higher level of variance for any time $t$. 
\\
\\
To have a better understanding of Equation \ref{Eq:ExpectedValue}, in Figure \ref{Fig:Mesh1} we plot the percentage profits obtained from our strategy with respect to different values of $\alpha$ and $\sigma$. We observe that the rate of increase in profits is higher with respect to an increase in the $\alpha$ of the stock when the volatility is high. This is in line with our intuition since at low levels of volatility, the stock price paths are already hitting the barrier without much deviation from their expected growth rates. 

\begin{figure}
\begin{center}
\caption{Expected profits as a function of $\alpha$ and $\sigma$ of the stock (assuming $\alpha>r_f=0.05$, $S_0=1$, $B=1.3$)}
\includegraphics[
height=2.50in,
width=5.0in
]%
{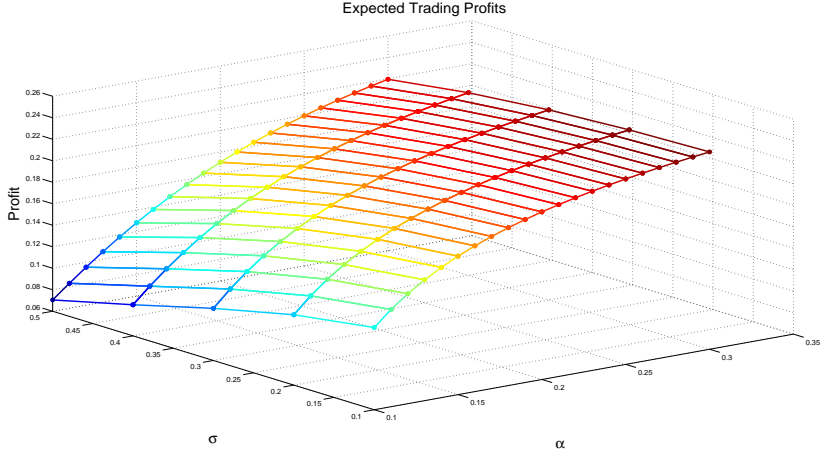}%
\label{Fig:Mesh1}%
\end{center}
\end{figure}

It is also clear that since for sufficiently large $t$, $P(\tau_B<\infty)=1$, then the variance is going to be bounded for the strategy that almost surely terminates in finite time. Therefore, we conclude that $\lim_{t\to\infty}{var(v(t))/t}=0$.\\
\\
\textbf{Probability of loss:}\\
The limit of the probability of loss is given by 
\begin{equation}\label{Eq:ProbLossEx2}
\lim_{t\to\infty}P(v(t)<0) =P(\tau_m > \sigma m/r_f) = 1-F_{IG}(\sigma m/r_f)
\end{equation}
where $F_{IG}$ is the cdf of the Inverse Gaussian distribution. Note that the probability of loss does not decay to zero in this strategy, which violated the definition of statistical arbitrage. Next, we present the results from our Monte Carlo experiment to verify our conclusions in this example. \\
\\
\textbf{Monte Carlo experiment for Example \ref{Ex:2}:}\\
Let's consider the strategy described in Example 2, where $\alpha=0.16$, $r_f=0.04$, and volatility $\sigma=0.2$. We borrow from the risk free rate and long one unit of stock at time 0. We set $S_0=1$ and the barrier equals to $1.2$. We simulate 10,000 paths of the the daily stock prices with the number of time steps $M=252$ (i.e. $\Delta t=1/252$). We implement our trading strategy terminating whenever we hit the barrier \$1.2 and investing all immediately to the money market account. Therefore, once the stock price hits the barrier, the variance becomes zero and our profits grow at the risk free rate of $r_f$.

The empirical distribution of discounted cumulative profits can be seen in Figure \ref{Fig:1} for investment horizons of one, two, five, ten, twenty, and fifty years. The empirical distribution converges with a bounded variance, but the probability of loss does not decay to zero as we expect. As presented in Figure \ref{Fig:2}, Monte Carlo simulation results are consistent with the theoretical results and the time averaged variance decays to zero as expected. 

\begin{figure}
\begin{center}
\caption{Monte Carlo simulation of discounted cumulative trading profits of statistical arbitrage strategy with respect to the time strategy is implemented. Parameters: $S_0=1$, $B=1.2$, $\alpha=0.16$, $r_f=0.04$, $\sigma=0.2$, $N=10000$, $M=252$}
\includegraphics[
height=2.50in,
width=5.20in
]%
{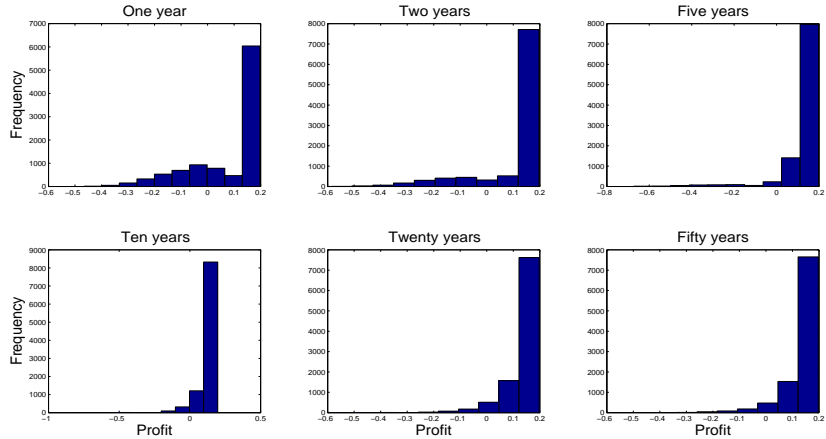}%
\label{Fig:1}%
\end{center}
\end{figure}

\begin{figure}
\begin{center}
\caption{Evolution of mean, time averaged variance, and probability of loss for the given trading strategy in Example \ref{Ex:2}.}
\includegraphics[
height=1.50in,
width=5.20in
]%
{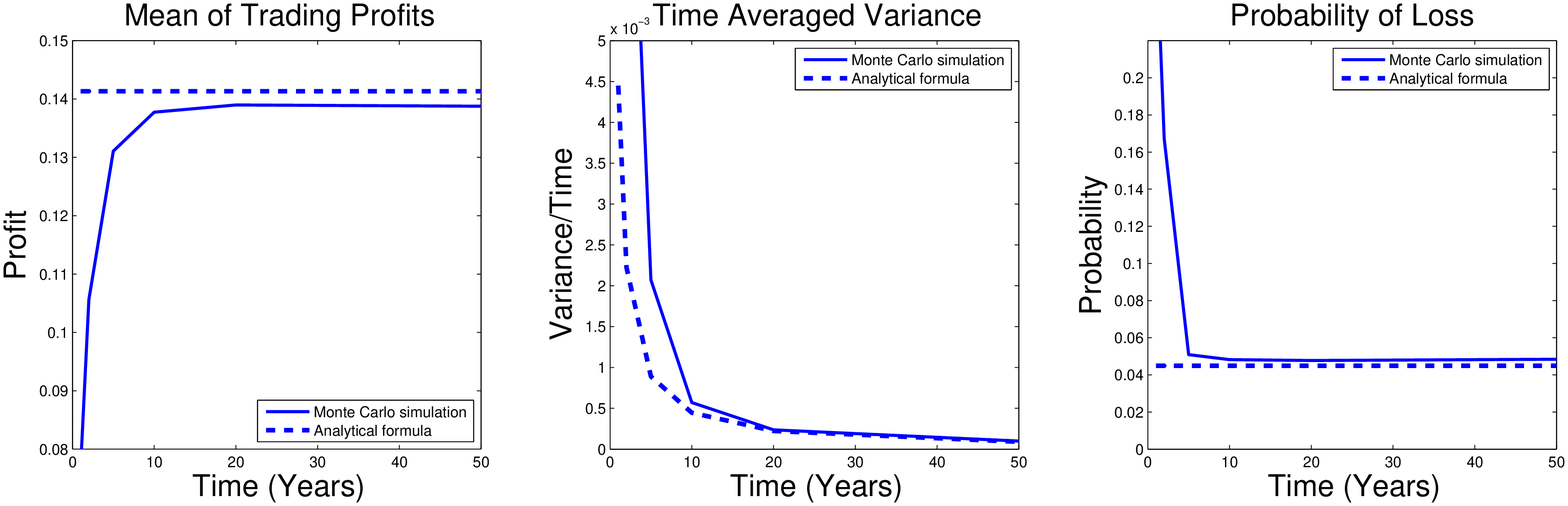}%
\label{Fig:2}%
\end{center}
\end{figure}

\begin{figure}
\begin{center}
\caption{First passage time density $\tau_B$ with respect to time}
\includegraphics[
height=1.60in,
width=4.0in
]%
{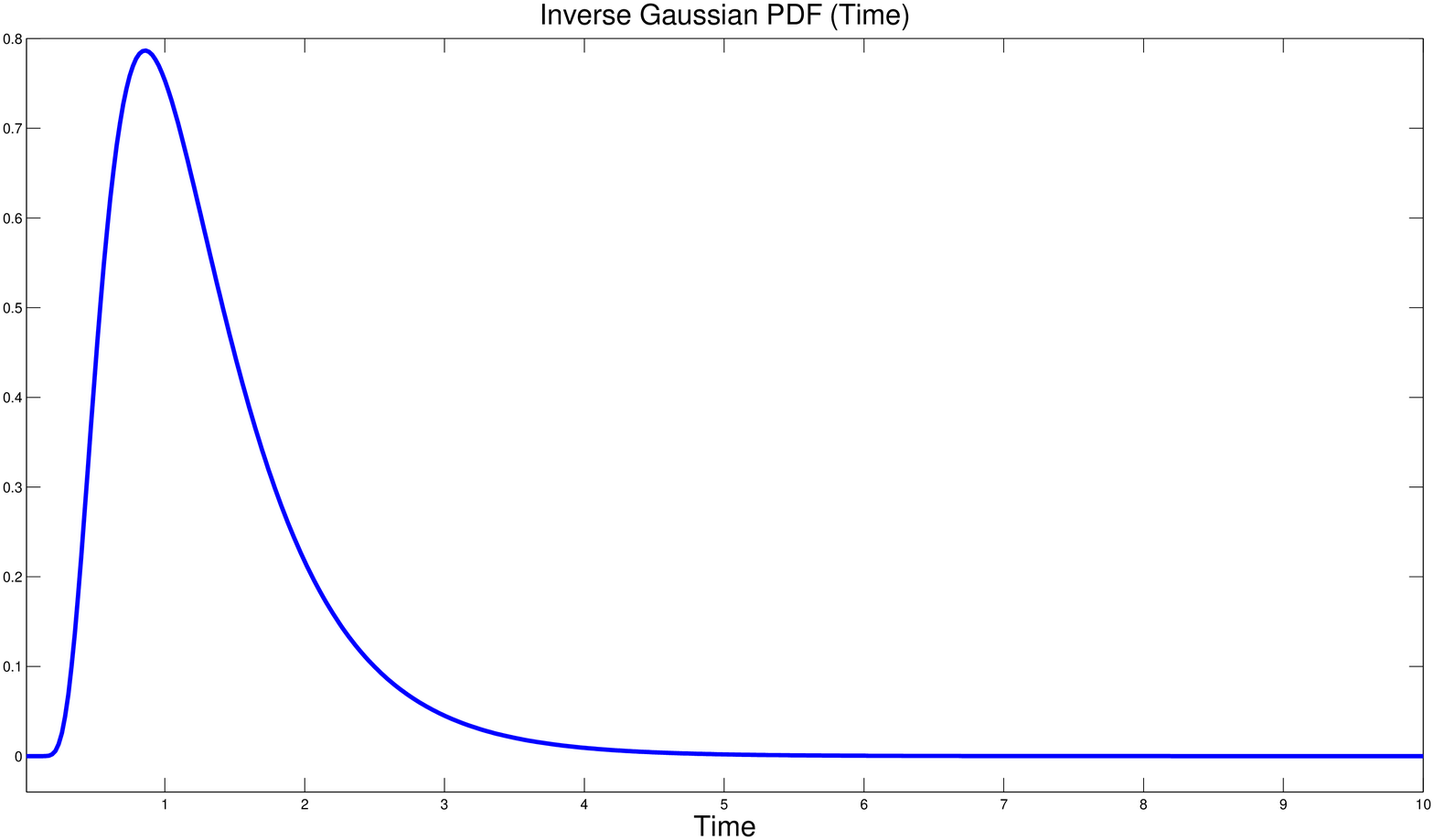}%
\label{Fig:3}%
\end{center}
\end{figure}

Figure \ref{Fig:3} plots the density of first passage time. We observe that the first passage time density decays to zero rapidly as a function of time, and our trading strategy terminates with very high probability for sufficiently large investment horizons. However, in this example there always exists stock price paths that hits the barrier too late to yield positive profit. In our next example we modify the trading strategy to obtain convergence to zero probability of loss.\\
\\
\textbf{Example 3:}\label{Ex:3} (``Buy and Hold Until Barrier'')\\
Different from the previous example, in this example we utilize a deterministic stopping boundary and show that statistical arbitrage can be obtained. Our trading strategy is as follows: at time 0 we long one unit of stock by borrowing from the bank. If the stock price hits $S_0(1+k)e^{r_f t}$ we sell, realizing the profit of $k$, and invest immediately in the money market account. This strategy is demonstrated in Figure \ref{Fig:TradingRule}, where we simulate 10 daily stock price paths for one year, and $k$ is $0.05$. 
 
\begin{figure}
\begin{center}
\caption{Demonstration of the trading rule in Example 3 with simulated paths of geometric Brownian motion. Sell if the stock price hits to the stopping boundary (Parameters: $r_f=0.04$, $\alpha=0.16$, $\sigma=0.2$)}
\includegraphics[
height=2.00in,
width=5.0in
]%
{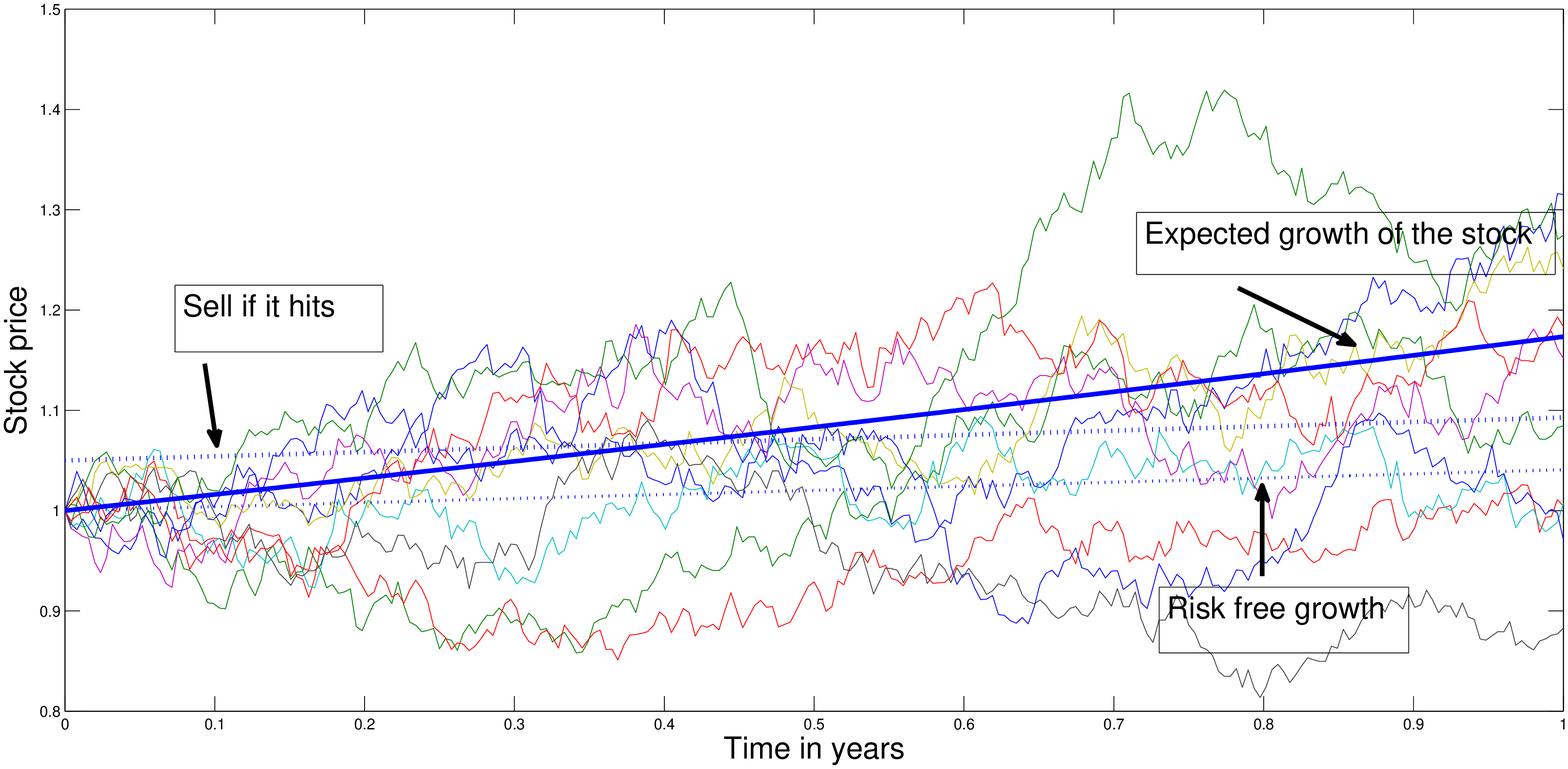}%
\label{Fig:TradingRule}%
\end{center}
\end{figure}

The discounted cumulative trading profits from this strategy can be written as 
\begin{equation}
v(t) = \begin{cases}
S_0 k & \hbox{ if } \tau_{k} \in [0,t]  \\
S_te^{-r_f t} - S_0 & \hbox{ else } 
 \end{cases}, 
\end{equation}
where 
\begin{equation}
\tau_k = \min\{t \geq 0: S_t = S_0 (1+k) e^{r_f t}\}.
\end{equation}

If the stock price hits the barrier level $B=S_0(1+k)e^{r_f t}$, then we can equivalently write this event as the Brownian motion with drift hitting $\ln(1+k)/\sigma$. The barrier level $k^*$ for the Brownian motion with drift $X_t$ is given as $k^*=\ln(1+k)/\sigma$, since $\ln(e^{-r_f t}S_t/S_0)/\sigma = \ln(1+k)/\sigma$. In this example we consider the discounted stock price process and we can write
\begin{align}\label{Eq:BMwithDrift2}
X_t = & \mu t + W_t\\\label{Eq:BMwithDrift3}
\tau_{k^*} = & \min\{ t\geq 0: X_t = k^{*}\}
\end{align}
for $k^*=\ln(1+k)/\sigma>0$ and $\mu = (\alpha-r_f -\sigma^2/2)/\sigma$. 

Similar to the previous example, we conclude that the first passage time for the Brownian motion with drift is Inverse Gaussian distributed as 
\begin{equation}\label{Eq:InverseGaussian2}
\tau_{k^{*}}\sim \hbox{IG}\left( \frac{k^*}{\mu},{k^*}^2\right). 
\end{equation} \\
\\
\textbf{Expected value and variance of trading profits: }\\
Given $X_t$ with $\mu>0$, we have $P(\tau_{k^*}<\infty)=1$, and for sufficiently large $t$, the stock price path hits to the deterministic barrier. Then $\lim_{t\to\infty}E[v(t)] = E[\lim_{t\to\infty} v(t)] = S_0 k > 0$. Note that the boundedness of $v(t)$ and $\lim_{t\to\infty}v(t)=S_0k$ implies $\lim_{t\to\infty}var(v(t))/t=0$. In this trading strategy the holding of the risky asset becomes zero for sufficiently large $t$ and the variance decays to zero in time.\\
\\
\textbf{Probability of loss:}\\
Let $M(t)$ be the maximum of the process $X_t$ in the time interval $[0,t]$ and the probability that the maximum is less than the barrier $k^*=\ln(1+k)/\sigma$ is denoted by $P(M(t)<k^*)$. In our trading strategy the probability of loss $P(v(t)<0)$ is given by 
\begin{align}\label{Eq:ProbLoss}
P(v(t)<0) =& P(S_t < S_0 e^{r_ft},\tau_{k^*}>t) = P(S_t < S_0e^{r_f t}|\tau_{k^*}>t)P(\tau_{k^*}>t) \\\nonumber
=& P((\alpha - r_f - \sigma^2/2)t + \sigma W_t < 0| M(t)<k^*) P(M(t) < k^*) \\\nonumber
\equiv & P(Z<-(\alpha-r_f-\sigma^2/2)\sqrt{t}/\sigma|\tau_{k^*}>t) \\\nonumber
& \times \left(\Phi(\frac{k^*-(\alpha-r_f-\sigma^2/2)t/\sigma}{\sqrt{t}}) - e^{2k^*(\alpha-r_f-\sigma^2/2)/\sigma}\Phi(\frac{-k^*-(\alpha-r_f-\sigma^2/2)t/\sigma}{\sqrt{t}}) \right),\\\nonumber
\end{align}
where $Z$ is the standard normal random variable, and $\Phi(.)$ is the standard normal cdf. For $\alpha-r_f\geq\sigma^2/2$ as $t\to\infty$ the probability of loss goes to zero, whereas for $\alpha-r_f<\sigma^2/2$ we almost surely make a loss as $t\to\infty$. The probability of loss does not decay to zero for $0<\alpha-r_f <\sigma^2/2$.  

In the case of $0<\alpha-r_f < \sigma^2/2$, as $t\to\infty$ from Equation \ref{Eq:ProbLoss} we obtain 
\begin{equation}\label{Eq:ProbLossLimit}
\lim_{t\to\infty}P(v(t)<0) = 1- e^{2k^*(\alpha-r_f-\sigma^2/2)/\sigma} > 0,
\end{equation}
which means that for $0<\alpha-r_f < \sigma^2/2$ the probability of loss does not converge to zero under the hold until barrier strategy. 

Finite first passage time of the Brownian motion with drift in the case of $\alpha-r_f>\sigma^2/2$,  implies that there always exists a sufficiently large $T$ such that $v(t)=S_0 k$ for all $t\geq T$, which implies that $\lim_{t\to\infty} P(v(t)<0)=0$. For $\alpha-r_f>\sigma^2/2$, for sufficiently large $t$ the variance becomes zero. Therefore, we conclude that there exists statistical arbitrage opportunities in the Black-Scholes framework.

Next, we demonstrate the existence of statistical arbitrage via a Monte Carlo experiment.\\
\\
\textbf{Monce Carlo experiment for Example 3:}\\
To verify the validity and convergence of our statistical arbitrage strategy in Example 3, we present the results of a Monte Carlo experiment. We simulate 10,000 sample paths with daily time steps, i.e. $M=252$ for different investment horizons of $T=1,2,5,10,20,50$ years. We set $S_0=1$, $k=0.05$, $\alpha=0.16$, $r_f=0.04$, and $\sigma=0.2$.

\begin{figure}
\begin{center}
\caption{Evolution of the empirical distribution of discounted cumulative trading profits obtained from the trading strategy given in Example 3.}
\includegraphics[
height=2.50in,
width=5.20in
]%
{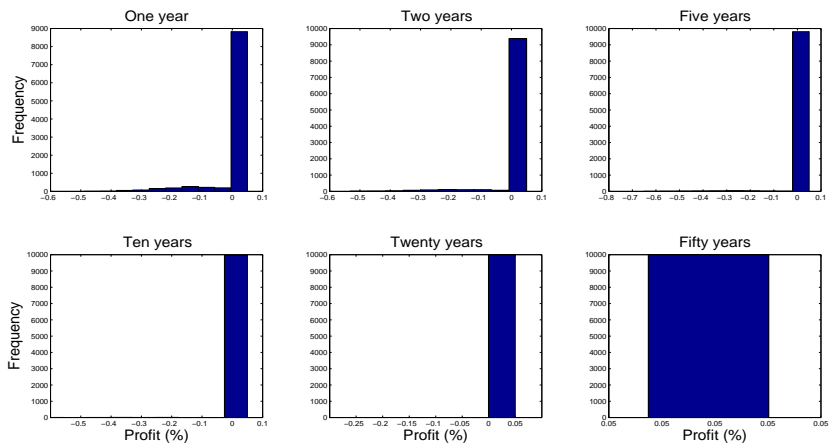}%
\label{Fig:Hist}%
\end{center}
\end{figure}

As can be seen in Figure \ref{Fig:TradingRule}, whenever a simulated stock price path hits the deterministic barrier of $S_0(1+k)e^{r_f t}$, we sell the stock and invest all to the money market account. For sufficiently large $t$, the average of the discounted cumulative trading profits becomes a point mass at $E(v(t))=k$, which can be seen in Figure \ref{Fig:Hist}. 

\begin{figure}
\begin{center}
\caption{Evolution of mean, time averaged variance, and probability of loss for the given trading strategy in Example 3.}
\includegraphics[
height=1.40in,
width=4.0in
]%
{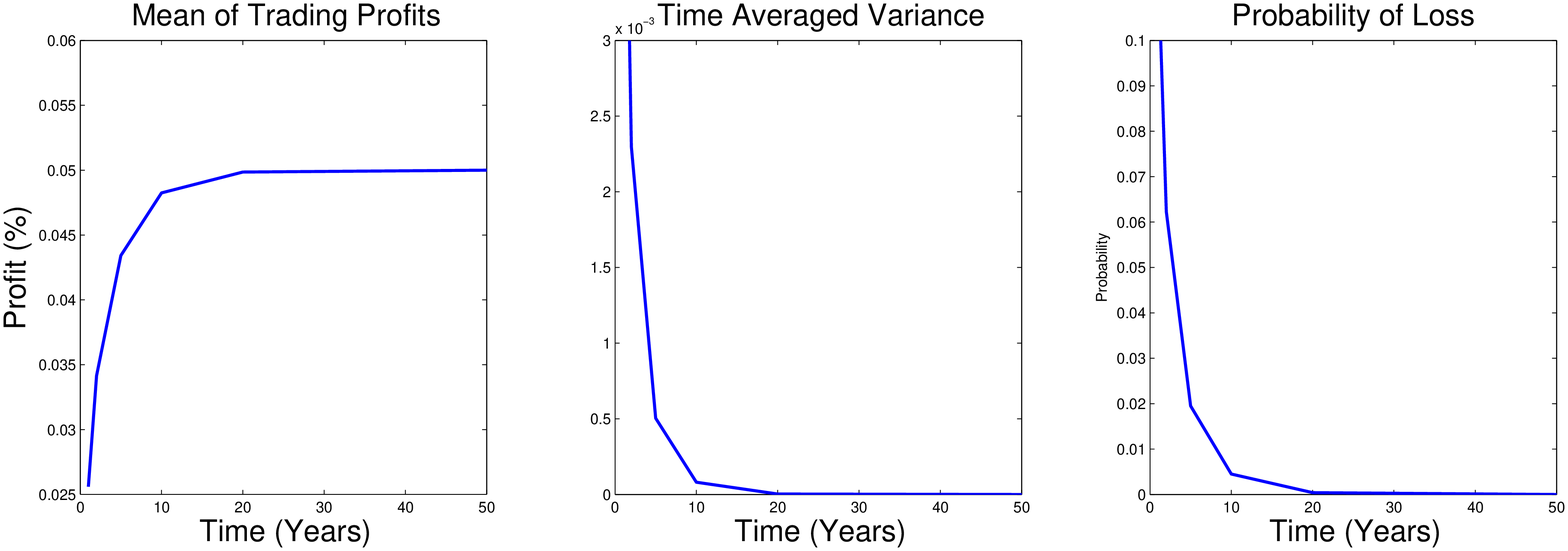}%
\label{Fig:MeanVar2}%
\end{center}
\end{figure}

In Figure \ref{Fig:MeanVar2}, we verify that for sufficiently large $t$ expected discounted profits converges to $k$, whereas the time averaged variance and the probability of loss both decay to zero. Therefore, Monte Carlo results are consistent with our theoretical results showing that there exists statistical arbitrage opportunities in the Black-Scholes framework. 

In the next example, we consider the case when the investor has the knowledge that a given stock will under-perform with low expected growth rate. In this case, there also exists statistical arbitrage opportunities. \\
\\
\textbf{Example: 4}\label{Ex:4} (``Short until barrier'')\\
If $\alpha<r_f$, then we can utilize a similar strategy as in Example 3, but this time we short the stock at time 0 and invest in the money market account at the risk free rate $r_f$. We close the short position whenever the stock price hits the boundary level, $S_0(1+k)^{-1}e^{r_f t}$. \footnote{Since short positions need to be closed in relatively short periods of time, the investor can close his short position at every $\Delta t$ time increment and re-open a new short position immediately, which does not affect our results in the absence of transaction costs.}

In this trading strategy the discounted cumulative trading profits from our strategy can be written as 
\begin{equation}
v(t) = \begin{cases}
S_0 k/(k+1) & \hbox{ if } \tau_k \in [0,t] \\
S_0 - S_te^{-r_f t} & \hbox{ else },
\end{cases}
\end{equation}
where $\tau_{k} = \min\{ t \geq 0: S_t = S_0 (1+k)^{-1}e^{r_f t}\}$, and the barrier level is $B = S_0 (1+k)^{-1}e^{r_f t}$. This is equivalent to the hitting time of the Brownian motion with drift
$$\tau_{-k^*}=\min(t\geq 0: -X_t = -k^*),$$
where $k^*=\ln(1+k)/\sigma$. Let $\mu=\frac{\sigma^2/2 - (\alpha-r_f)}{\sigma}$ to obtain $-X_t = -\mu t -W_t = -\mu t + W_t=(\alpha-r_f-\sigma^2/2)/\sigma + W_t$. Hence previous results can be applied for $\alpha-r_f<\sigma^2/2$. Therefore, our trading strategy ``short until barrier'' satisfies the statistical arbitrage condition, since we have $P(\tau_{-k^*}<\infty)=1$. 

\begin{figure}
\begin{center}
\caption{Evolution of the empirical distribution of discounted cumulative trading profits obtained from the trading strategy given in Example 4.}
\includegraphics[
height=2.50in,
width=5.20in
]%
{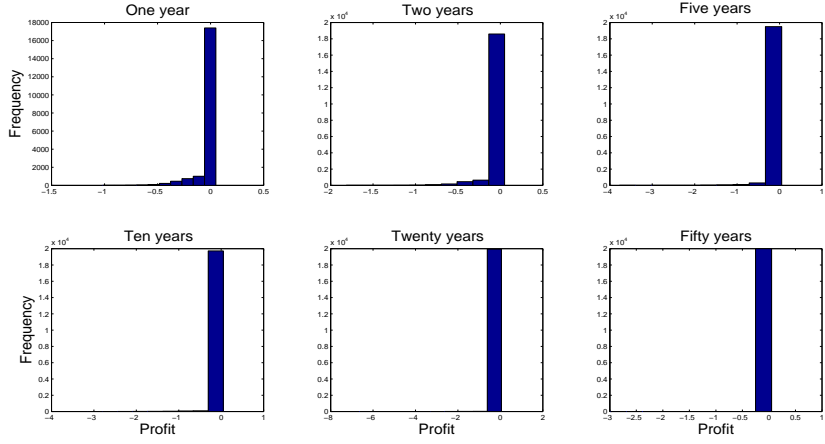}%
\label{Fig:ShortSellHist}%
\end{center}
\end{figure}

Alternatively, without considering a stopping boundary one can simply keep shorting the stock and invest the proceeds in the money market account. Since the discounted stock price decays to zero as given in Proposition \ref{Proposition1}, the variance also decays to zero, while $\lim_{t\to\infty}{E[v(t)]}=S_0$.\\
\\
\textbf{Monte Carlo experiment for Example 4: }\\
We consider the short selling strategy introduced in Example 4 with the following set of parameters: $\alpha=0.01$, $r_f=0.05$, $\sigma=0.2$, $N=10000$, $M=252$, and $k=0.05$. We simulate the stock price paths and whenever the stock price hits the barrier level of $S_0(1+k)^{-1}e^{r_f t}$ we close the short position.

\begin{figure}
\begin{center}
\caption{Evolution of mean, time averaged variance, and probability of loss for the given trading strategy in Example 4.}
\includegraphics[
height=1.50in,
width=4.50in
]%
{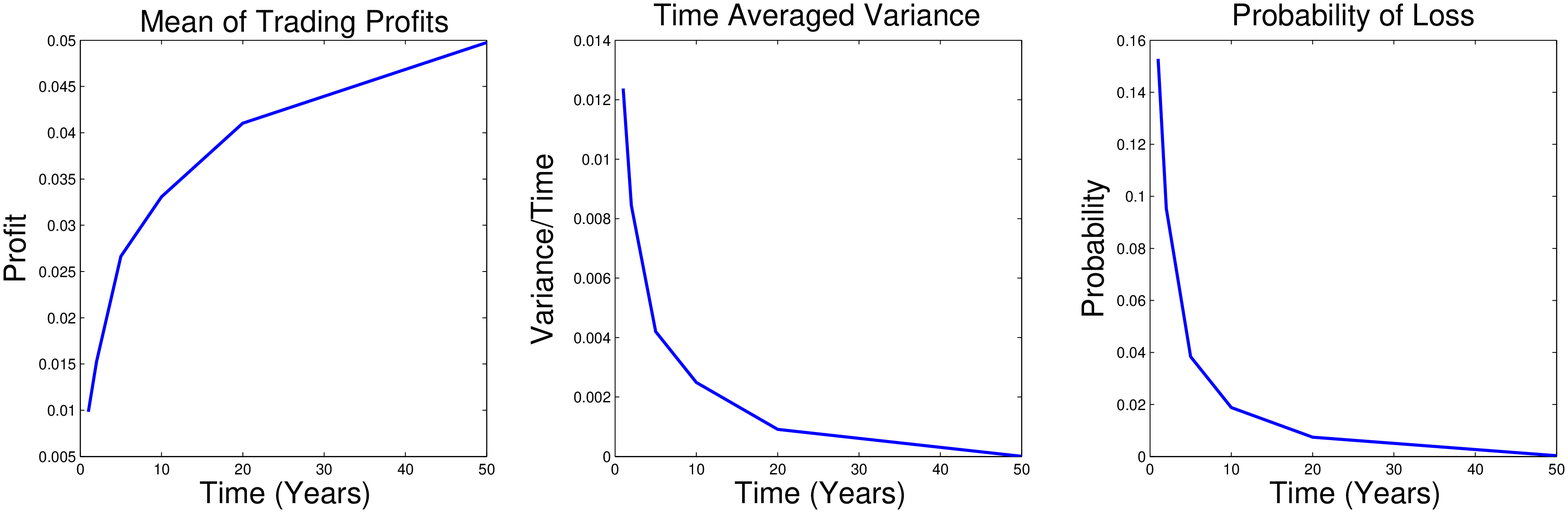}%
\label{Fig:ShortSellEvolution}%
\end{center}
\end{figure}

In Figure \ref{Fig:ShortSellHist} the time evolution of the histogram of the trading profits shows that the distribution of the trading profits converge to a point mass at the limiting trading profit $S_0[1-1/(1+k)]=S_0 k/(1+k)=0.0496$ as $t\to\infty$. 

Figure \ref{Fig:ShortSellEvolution} clearly shows that the expected value of the discounted trading profits is converging to $S_0 k/(1+k)$, while the probability of loss is decaying to zero. The time averaged variance decays to zero as required in the definition of statistical arbitrage.

In the next section we present the conditions that guarantee the existence of statistical arbitrage in the Black Scholes model.

\section{Existence of Statistical Arbitrage}\label{Section:StatArbBSModel}
In this section we present the condition that guarantees the existence of statistical arbitrage opportunities in the Black-Scholes framework. 

In the next theorem we present the case in which the stock price process almost surely hits the barrier in finite time. 
\begin{theorem}\label{Theorem1} Assume that the stock prices follow geometric Brownian motion given by
\begin{equation}
S_t = S_0 \exp\left( (\alpha-\sigma^2/2)t + \sigma W_t  \right), \quad \sigma>0.
\end{equation}
We define a deterministic stopping boundary $B=S_0(1+k)e^{r_f t}$, where $\alpha, r_f, k>0$ are constants and the first passage time is denoted as  
\begin{equation}\label{Eq:StoppingTime}
\tau_B = \min\left(t\geq 0: S_t = B\right).
\end{equation} 
Then the first passage time of the stock price process is finite almost surely, i.e. $P(\tau_B<\infty) = 1$ for $\alpha-r_f>\frac{\sigma^2}{2}$. 
\end{theorem}
%\begin{proof}
%Proof follows directly from Equations \ref{Eq:ShreveFormula}, \ref{Eq:BMwithDrift2}, and \ref{Eq:BMwithDrift3}.
%\end{proof}\\
\begin{proof}
Let's define a Brownian motion process with drift as follows: 
\begin{equation}
X_t = \frac{\ln(e^{-r_f t}S_t/S_0)}{\sigma} = \underbrace{\frac{(\alpha-r_f -\sigma^2/2)}{\sigma}}_{\mu}t + W_t.
\end{equation}
Note that  $S_t=B=S_0(1+k)e^{r_f t}$ if and only if $X_t = \ln(1+k)/\sigma$. Let $k^*=\ln(1+k)/\sigma>0$ with stopping time $\tau_{k^*}=\min(t\geq 0: X_t = k^*)$.

Therefore, $P(\tau_B<\infty)=1$ if and only if $P(\tau_{k^*}<\infty)=1.$ Following a procedure that is similar as given in \cite{Shreve:2004} (see page 120) we introduce an exponential martingale process $Z(t)$ given by
\begin{equation}\label{Eq:Martingale1}
Z(t) = \exp(\theta X(t) -(\theta \mu + \frac{\theta^2}{2})t ),
\end{equation}
where $Z(0)=1$ and $\theta$ is an arbitrary non-negative constant. Since any stopping martingale is still a martingale, we have
\begin{equation}\label{Eq:Martingale2}
E\left[ \exp\left( \theta X(t\wedge \tau_{k^*} ) - (\theta \mu +\frac{\theta^2}{2}) (t\wedge \tau_{k^*}) \right) \right]=1.
\end{equation}

In the above equation, if $\tau_{k^*}=\infty$ the term $\exp( -(\theta \mu + \frac{\theta^2}{2})(t\wedge \tau_{k^*}))$ goes to zero, whereas if $\tau_{k^*}<\infty$, we have $\exp( -(\theta\mu +\frac{\theta^2}{2})(t\wedge \tau_{k^*}) )=\exp( -(\theta \mu +\frac{\theta^2}{2})\tau_{k^*})$ for sufficiently large $t$. 

The other term $\exp(\theta X(t\wedge \tau_{k^*}))$ is always bounded by $\exp(\theta k^*)$ if $\tau_{k^*}=\infty$. If $\tau_{k^*}<\infty$, this term equals to $\exp(\theta W(t\wedge\tau_{k^*}))=\exp(\theta k^*)$. 
The product of two exponential terms can be captured by 
\begin{align}
\lim_{t\to\infty}\exp\left(\theta X(t\wedge\tau_{k^*}) -(\theta\mu +\frac{\theta^2}{2})(t\wedge\tau_{k^*}) \right)=\mathbf{1}_{\{\tau_{k^*}<\infty\}} \exp\left( \theta k^* - (\theta\mu +\frac{\theta^2}{2})\tau_{k^*}\right),
\end{align}
where $\mathbf{1}_{\{\tau_{k^*}<\infty\}}=\begin{cases}
1 & \hbox{ if } \tau_{k^*}<\infty \\
0 & \hbox{ if } \tau_{k^*}=\infty
\end{cases}.$

Taking the limit in Equation \ref{Eq:Martingale2} and interchanging the limit and expectation as a result of the dominated convergence theorem, we obtain:
\begin{align}\label{Eq:ThmResult1}
E\left[\mathbf{1}_{\{\tau_{k^*}<\infty\}} \exp(\theta k^* - (\theta\mu+\theta^2/2)\tau_{k^*})\right] &=  1 \\\nonumber
E[\mathbf{1}_{\{\tau_{k^*}<\infty\}} e^{- (\theta \mu+\theta^2/2) \tau_{k^*}}] & = e^{-\theta k^*},
\end{align}
which holds for $(\theta \mu + \theta^2/2)>0$ and $\theta>0$. 

For the case $\mu>0$ (i.e. $\alpha-r_f>\sigma^2/2$), we can take the limit on both sides in Equation \ref{Eq:ThmResult1} for $\theta\downarrow 0$ which yields $P(\tau_{k^*}<\infty)=1$. However, for the case $\mu<0$ and $\mu>-\theta/2$, $\theta$ can only converge to the positive constant, $\theta\downarrow -2\mu$, for which we obtain
\begin{equation}
E[\mathbf{1}_{\{\tau_{k^*}<\infty\}}] = e^{ 2\mu k^*}<1, \quad \hbox{ for } \mu<0
\end{equation}
and therefore $P(\tau_{k^*}<\infty)<1$. 
\end{proof}\\
\\
\textbf{Monte Carlo experiment for the case: $0<\alpha-r_f<\sigma^2/2$}\\
We demonstrate that we do not obtain statistical arbitrage for $0<\alpha-r_f<\sigma^2/2$ via buy and hold (long) until barrier strategies. Consider the parameters given by $\alpha=0.05$, $r_f=0.04$, $k=0.05$, $\sigma=0.2$, $N=10000$, $M=252$.

\begin{figure}
\begin{center}
\caption{Evolution of mean, time averaged variance, and probability of loss for the buy and hold (long) until barrier trading strategy for the case of $0\leq \alpha-r_f\leq \frac{\sigma^2}{2}$. Parameters given as $\alpha=0.05$, $r_f=0.04$, $k=0.05$, $\sigma=0.2$, $N=10000$, $M=252$.}
\includegraphics[
height=1.50in,
width=4.50in
]%
{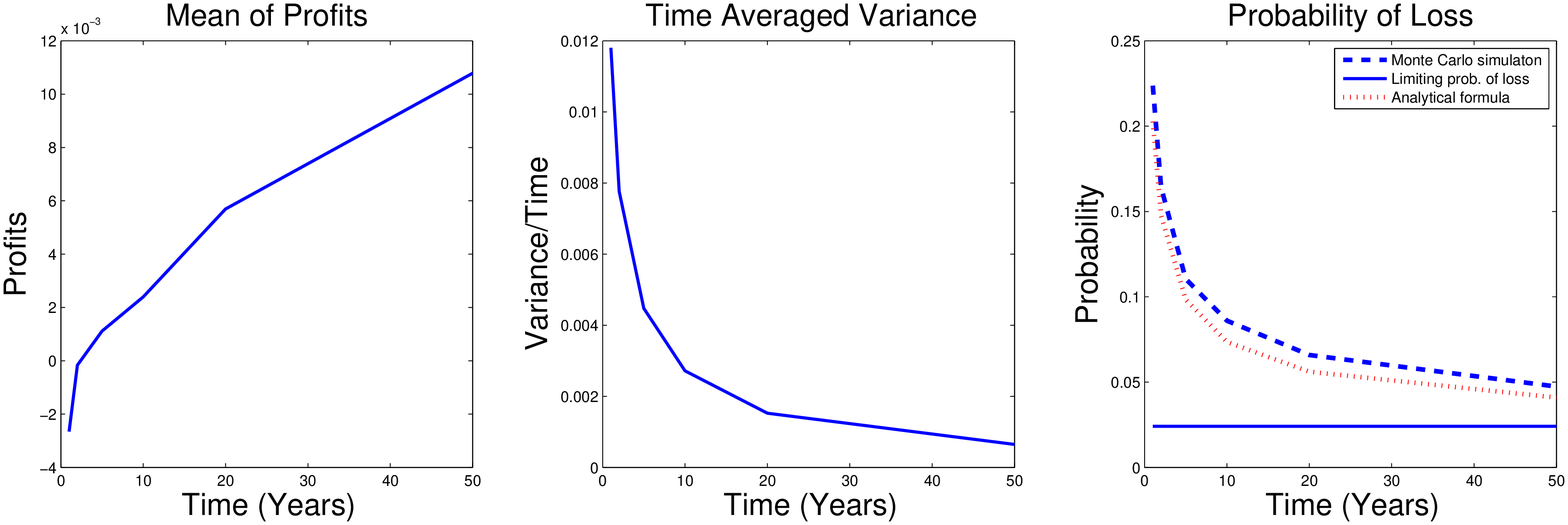}%
\label{Fig:NoStatArb}%
\end{center}
\end{figure}

In Figure \ref{Fig:NoStatArb} we plot the time evolution of mean, time averaged variance and the probability of loss for the trading strategy considered. We observe that the probability of loss obtained from the analytical formula given in Equation \ref{Eq:ProbLoss} and the Monte Carlo estimator are quite close to each other. We also plot the limiting probability of loss when $T\to\infty$, as given by the analytical formula in Equation \ref{Eq:ProbLossLimit}. In this example statistical arbitrage is not obtained simply because for the long until barrier strategy the probability of loss decays to zero only if we have $\alpha-r_f>\sigma^2/2$ as given in Equation \ref{Eq:ProbLoss}. Therefore, it is clear that Condition 3 in Definition \ref{Def:StatisticalArbitrage} is not satisfied. 

The next corollary states the symmetric result for $\alpha-r_f<\sigma^2/2$.

\begin{corollary}\label{Corollary1}
The result obtained in Theorem \ref{Theorem1} applies for the symmetric case when the stopping boundary is defined as $B = S_0 (1+k)^{-1}e^{r_f t}$, (with an abuse of notation we still denote the barrier with $B$)
\begin{equation}
\tau_{B} = \min\left(t\geq 0: S_t = B\right).
\end{equation} 
Then the first passage time of the stock price process to level $B$ is finite almost surely, i.e. $P(\tau_B<\infty) = 1$, for $\alpha-r_f < \sigma^2/2$ and $P(\tau_B<\infty)<1$ for $\alpha-r_f\geq \sigma^2/2$.  
\end{corollary}
\begin{proof}
Consider $X_t = \ln(S_t e^{-r_f t}/S_0)$ and $X_t = -\mu t + W_t$, where $\mu=\frac{r_f-\alpha+\sigma^2/2}{\sigma}.$ Then equivalently we can write $X_t = -\mu t -W_t$. Let $k^*=-\ln(1+k)/\sigma<0$ with stopping time $\tau_{k^*}=\min(t\geq 0: X_t = k^*)$.

Following the similar arguments as in \cite{Shreve:2004} page 110, we consider the following exponential martingale
\begin{equation}
Z(t) = \exp\left( -\theta X_t - (\mu\theta + \theta^2/2)t\right), 
\end{equation}
where $\theta>0$ is an arbitrary constant. 

The term $\exp(-\theta X_t)$ is always bounded by $e^{\theta k^*}$. We obtain
\begin{equation}
E[ e^{-(\mu \theta + \theta^2/2)\tau_{k^*}} \mathbf{1}_{\{\tau_{k^*}<\infty\}}]=e^{-k^* \theta}.
\end{equation}
There are two cases to consider: (i) $\mu>0$; (ii) $\mu<0$ and $\theta>-2\mu$. In the first case, we can let $\theta\to 0$, then obtain $P(\tau_{k^*}<\infty)=1$ for $\alpha-r_f<\sigma^2/2$. In the second case, $\mu<0$ and $\theta>-2\mu$, we have $\alpha-r_f\geq \sigma^2/2$ and $\theta$ converges to a positive constant, and thus we have $P(\tau_{k^*}<\infty)<1$.
\end{proof}

\begin{theorem}\label{Theorem:NoStatArbCondition}
In the Black-Scholes model, there exists statistical arbitrage in the sense of Definition \ref{Def:StatisticalArbitrage} if $\alpha-r_f \neq \sigma^2/2$. If $\alpha-r_f>\sigma^2/2$, then there exists statistical arbitrage for the long until barrier strategies, whereas if $\alpha-r_f<\sigma^2/2$, there exists statistical arbitrage for the short until barrier strategies.  
\end{theorem}
\begin{proof}
First consider the case $\alpha-r_f>\frac{\sigma^2}{2}$. We construct our ``long until barrier'' type of trading strategy as follows. Long the stock at time 0 by borrowing from the bank at the interest rate $r_f$, and hold the stock until it hits the barrier and sell it at the level $S_0(1+k)e^{r_f t}$. Then, by Theorem \ref{Theorem1}, we have $P(\tau_{k^*}<\infty)=1$. As we have discussed in Example 3, for sufficiently large $t$, we have the discounted trading profits given as 
\begin{equation}
v(t) = \begin{cases}
S_0 k & \hbox{ if } \tau_{k^*} \in [0,t],\\
S_t e^{-r_f t} - S_0 & \hbox{ else }.  
\end{cases}
\end{equation} 
Then, we have $\lim_{t\to\infty}E[v(t)]=S_0 k>0$, $\lim_{t\to\infty}var(v(t))/t=0$, and since for sufficiently large $t$ stock price process almost surely hits the barrier level, the probability of loss decays to zero, i.e. $\lim_{t\to\infty}{P(v(t)<0)}=0$. Therefore, Definition \ref{Def:StatisticalArbitrage} is satisfied. 

For second the case, $\alpha-r_f<\sigma^2/2$, we can consider ``short until barrier'' type of trading strategy. At time 0 short one unit of stock $S_0$ and invest proceedings in the money market account. As we analysed in Example 4, the cumulative discounted trading profits are given as 
\begin{equation}
v(t) = \begin{cases}
S_0 k/(k+1) & \hbox{ if } \tau_{k^*}\in [0,t],\\
S_0 - S_t e^{-r_f t} & \hbox{ else }.
\end{cases}
\end{equation}
Similarly, we have  $\lim_{t\to\infty}E[v(t)]=S_0 k/(k+1)>0$, $\lim_{t\to\infty}var(v(t))/t=0$, and by Corollary \ref{Corollary1} stock price paths almost surely hit the barrier in finite time, i.e. $P(\tau_{k^*}<\infty)=1$. The probability of loss decays to zero while the mean of trading profits converge to $S_0 k/(1+k)$. Therefore, for $\alpha-r_f>\sigma^2/2$ and $\alpha-r_f<\sigma^2/2$, we are able to obtain statistical arbitrage via long and short until barrier strategies, respectively.
\end{proof}

\section{On the Definition of Statistical Arbitrage}\label{Section:OtherProperties}
In this section, we prove additional properties of the statistical arbitrage strategies characterized by Definition \ref{Def:StatisticalArbitrage}. In the next proposition we prove that if the variance itself decays to zero in time, then this implies that the probability of loss decays to zero. Hence, Condition 3 of Definition \ref{Def:StatisticalArbitrage} becomes redundant. We also prove that if the expected trading profits goes to infinity in time and the variance converges to a constant, then this implies Condition 3 is satisfied. 

\begin{proposition}\label{PropositionStatArbDef1} Given the probability space $(\Omega,\mathcal{A},P)$ and stochastic process $\{v(t):t\geq 0\}$ defined on this space. Consider that we have the following properties for the trading strategy $\{v(t):t\geq 0\}$
\begin{enumerate}
\item $v(0) = 0$
\item $\lim_{t\to \infty} E[v(t)] > 0$,
\item $\lim_{t\to \infty}{var(v(t))}=0$, 
\end{enumerate}
then conditions $1-3$ implies $\lim_{t\to \infty}{{P}(v(t)<0)}=0$. 
\end{proposition}
\begin{proof}
Cantelli's inequality \cite{Cantelli:1910}, which is a single tail version of Chebyshev's inequality, states that for a real random variable $X$ with mean $\mu$ and variance $\sigma^2$
\begin{equation}
P(X  - \mu\geq a) \leq \frac{\sigma^2}{\sigma^2+a^2}, 
\end{equation}
where $a\geq 0$. We can change the sign of $X$ and consider $-X$ with mean $-\mu$ which yields
\begin{equation}
P(-X  + \mu\geq a)  = P(X\leq \mu-a)\leq \frac{\sigma^2}{\sigma^2+a^2}, 
\end{equation}

For each fix value of time $t$ we have a random variable $v(t)$ with mean $\mu_t$ and variance $\sigma^2_t$. Consider $P(v(t)\leq \mu_t - a)\leq \frac{\sigma^2_t}{\sigma^2_t+a^2}$
by setting $a= \mu_t$. By above Condition 2 we can always find a sufficiently large time $t_0$ such that $\forall t\geq t_0$, $\mu_t=E[v(t)]>0$, which implies $P(v(t)<0)\leq \frac{\sigma_t^2}{\sigma^2_t+\mu_t^2}$. Taking the limit on both sides we have $\lim_{t\to \infty}{P(v(t) < 0)} = 0$ as required.
\end{proof}
 
We prove a second property for the statistical arbitrage strategy $v(t)$ in the next proposition. 
\begin{proposition} Consider that we have the following properties for the trading strategy $\{v(t):t\geq 0\}$
\begin{enumerate}
\item $v(0) = 0$
\item $\lim_{t\to \infty} E[v(t)] = \infty$,
\item $\lim_{t\to \infty}{var(v(t))}= c$, 
\end{enumerate}
where $c$ is a positive constant. Then, conditions $1-3$ imply $\lim_{t\to \infty}{{P}(v(t)<0)}=0$. 
\end{proposition}
\begin{proof}
Proof is similar to the proof of Proposition \ref{PropositionStatArbDef1}. 
\end{proof}

At any initial time $t_0$ ($t\geq t_0$), let the collection of stochastic processes $v(t)$ satisfying Definition \ref{Def:StatisticalArbitrage} be denoted by $\mathcal{C}$. In the next proposition, we prove that $\mathcal{C}$ is a convex set.
\begin{proposition}
Given any two trading strategies $v_1(t)$, $v_2(t)\in \mathcal{C}$, their linear combination, $v^*(t)=av_1(t)+(1-a)v_2$, is also in $\mathcal{C}$. 
\end{proposition}
\begin{proof}
Let $v_1(t)$ and $v_2(t)$ be any two stochastic processes that satisfy Definition $\ref{Def:StatisticalArbitrage}$. Let $v^*=a v_1(t) + (1-a)v_2(t)$ where $a\in [0,1]$.
(i) Since both $v_1(0)=0$ and $v_2(0)=0$, then $v^*(0)=0$.

(ii) We have $a\lim_{t\to \infty}  E[v_1(t)] = \lim_{t\to \infty}E[a v_1(t)] > 0$ and $$(1-a)\lim_{t\to \infty} E[v_2(t)] = \lim_{t\to \infty} E[(1-a)v_2(t)] > 0,$$ which implies $\lim_{t\to \infty} E[av_1(t)+(1-a)v_2(t)] = \lim_{t\to \infty} E[v^*(t)] > 0$. 

(iii) We have $a \lim_{t\to \infty}{{P}(v_1(t)<0)}=\lim_{t\to \infty}{{P}(a v_1(t)<0)}=0$ and similarly $\lim_{t\to \infty}{{P}((1-a)v_2(t)<0)}=0$, which implies $\lim_{t\to \infty}{{P}(v^*(t)<0)}=0$.

(iv) $var(av_1(t)+(1-a)v_2(t))=a^2var(v_1(t)) + (1-a)^2var(v_2(t)) + 2a(1-a)cov(v_1(t),v_2(t))$ since $cov(v_1,v_2)=\rho \sigma_{1}\sigma_{2}$ with $\rho\in[0,1]$, we obtain $$\lim_{t\to \infty}{\frac{var(v^*(t))}{t}}=0$$ if ${P}(v^*(t)<0)>0,\quad$ $\forall t<\infty$.  
\end{proof}

As a result of the convexity of the set of statistical arbitrage trading strategies, we can consider the linear combination of two trading strategies and obtain the optimal investment weights that minimizes the variance. More generally, the mean-variance analysis of portfolio theory can be applied to obtain the efficient set of statistical arbitrage strategies that invests into a set of stocks that satisfy the statistical arbitrage condition we derived. 

\begin{remark}
We can minimize the variance of the linear combination of two statistical arbitrage strategies as 
\begin{equation}
\min_{a}{a^2\sigma_1^2 + (1-a)^2\sigma_2^2 + 2a(1-a)\rho \sigma_1 \sigma_2},
\end{equation}
where $a\in[0,1]$, $\sigma_1$ and $\sigma_2$ is the standard deviation of $v_1(t)$ and $v_2(t)$, respectively. The optimal portfolio weights that minimize the variance are given by 
\begin{equation}
\hat{a}= \frac{\sigma_2^2-\rho \sigma_1\sigma_2}{\sigma_1^2+\sigma_2^2-2\rho \sigma_1\sigma_2}, \quad 1-\hat{a}=\frac{\sigma_1^2-\rho\sigma_1\sigma_2}{\sigma_1^1+\sigma_2^2 - 2\rho\sigma_1\sigma_2}\hbox{ for time } t.
\end{equation} 
\end{remark}

\section{Conclusion}\label{Section:Conclusion}
The statistical arbitrage opportunities can be considered as riskless profit opportunities in the limit. The existence of statistical arbitrage opportunities and the admissible set of such trading opportunities in an economy are closely related to the inefficiency of the market. Whenever identified by traders, statistical arbitrage opportunities can be exploited and this helps the market to move towards efficiency. 

In this study, we derive the no-statistical arbitrage condition in the Black-Scholes model given by $0<\alpha-r_f<\sigma^2/2$, which implies that the Sharpe ratio of any given stock must be bounded by $\sigma/2$. We showed that if there are inefficiencies in the market then an investor can utilize statistical arbitrage opportunities in the Black-Scholes framework. We design trading strategies by introducing a stopping boundary that assures the existence of statistical arbitrage profits. 

There are various future research directions as a result of our study. First one can consider extensions of the Black-Scholes model and derive no-statistical arbitrage conditions for more general models. Furthermore, our results can be extended for the portfolios of stocks and the optimal statistical arbitrage strategies can be designed by minimizing the variance of the statistical arbitrage portfolios.  

%\section{Acknowledgements}

%\textbf{Example 1: }($\alpha-r_f > \frac{\sigma^2}{2}(1-\sigma)$)
%Long and hold until barrier. Set $\alpha=0.05$, $r_f=0.04$, $\sigma=0.15$, which should yield %statistical arbitrage according to our theorem.

%\begin{figure}
%\begin{center}
%\caption{Evolution of mean, time averaged variance, and probability of loss for the given trading strategy in the case $\alpha-r_f>\frac{\sigma^2}{2}(1-\sigma)$.}
%\includegraphics[
%height=1.50in,
%width=4.0in
%]%
%{NewEx1.eps}%
%\label{Fig:MeanVar2}%
%\end{center}
%\end{figure}

%\textbf{Example 2:}($\alpha-r_f <-\frac{\sigma^2}{2}(1-\sigma)$)
%Short until barrier


%\bibliographystyle{elsarticle-harv}

%\section{References}

\section{Acknowledgements}
I would like to thank Professor Giray Okten at Florida State University for his valuable comments that improved the manuscript.

\end{document}